\documentclass[pra,notitlepage,twocolumn,nofootinbib,superscriptaddress]{revtex4-2}
\usepackage{mathtools}
\usepackage{amsmath}
\usepackage[shortlabels]{enumitem}
\usepackage{graphicx,epic,eepic,epsfig,amsmath,latexsym,amssymb,verbatim, color}
\usepackage{amsfonts}       
\usepackage{nicefrac}       
\usepackage{mathrsfs}
\usepackage{bbm}
\usepackage{float}
\usepackage{tikz}
\usetikzlibrary{chains}
\usetikzlibrary{fit}
\usetikzlibrary{arrows}
\usetikzlibrary{decorations}

\usepackage{epsfig}
\usetikzlibrary{shapes.symbols,patterns} 
\usepackage{pgfplots}

\usepackage[strict]{changepage}
\usepackage{hyperref}
\hypersetup{colorlinks=true,citecolor=blue,linkcolor=blue,filecolor=blue,urlcolor=blue,breaklinks=true}

\usepackage[marginal]{footmisc}
\usepackage{url}
\usepackage{theorem}

\newtheorem{definition}{Definition}
\newtheorem{proposition}{Proposition}
\newtheorem{lemma}[proposition]{Lemma}

\newtheorem{theorem}[proposition]{Theorem}

\def\squareforqed{\hbox{\rlap{$\sqcap$}$\sqcup$}}
\def\qed{\ifmmode\squareforqed\else{\unskip\nobreak\hfil
\penalty50\hskip1em\null\nobreak\hfil\squareforqed
\parfillskip=0pt\finalhyphendemerits=0\endgraf}\fi}
\def\endenv{\ifmmode\;\else{\unskip\nobreak\hfil
\penalty50\hskip1em\null\nobreak\hfil\;
\parfillskip=0pt\finalhyphendemerits=0\endgraf}\fi}
\newenvironment{proof}{\noindent \textbf{{Proof~} }}{\hfill $\blacksquare$}

\newcounter{remark}
\newenvironment{remark}[1][]{\refstepcounter{remark}\par\medskip\noindent%
\textbf{Remark~\theremark #1} }{\medskip}

\newcounter{example}

\mathchardef\ordinarycolon\mathcode`\:
\mathcode`\:=\string"8000
\def\vcentcolon{\mathrel{\mathop\ordinarycolon}}
\begingroup \catcode`\:=\active
  \lowercase{\endgroup
  \let :\vcentcolon
  }

\usepackage{cleveref}
\usepackage{graphicx}
\usepackage{xcolor}

\RequirePackage[framemethod=default]{mdframed}
\newmdenv[skipabove=7pt,
skipbelow=7pt,
backgroundcolor=darkblue!15,
innerleftmargin=5pt,
innerrightmargin=5pt,
innertopmargin=5pt,
leftmargin=0cm,
rightmargin=0cm,
innerbottommargin=5pt,
linewidth=1pt]{tBox}

\newmdenv[skipabove=7pt,
skipbelow=7pt,
backgroundcolor=red!15,
innerleftmargin=5pt,
innerrightmargin=5pt,
innertopmargin=5pt,
leftmargin=0cm,
rightmargin=0cm,
innerbottommargin=5pt,
linewidth=1pt]{rBox}

\newmdenv[skipabove=7pt,
skipbelow=7pt,
backgroundcolor=blue2!25,
innerleftmargin=5pt,
innerrightmargin=5pt,
innertopmargin=5pt,
leftmargin=0cm,
rightmargin=0cm,
innerbottommargin=5pt,
linewidth=1pt]{dBox}
\newmdenv[skipabove=7pt,
skipbelow=7pt,
backgroundcolor=darkkblue!15,
innerleftmargin=5pt,
innerrightmargin=5pt,
innertopmargin=5pt,
leftmargin=0cm,
rightmargin=0cm,
innerbottommargin=5pt,
linewidth=1pt]{sBox}
\definecolor{darkblue}{RGB}{0,76,156}
\definecolor{darkkblue}{RGB}{0,0,153}
\definecolor{blue2}{RGB}{102,178,255}
\definecolor{darkred}{RGB}{195,0,0}

\newcommand{\nc}{\newcommand}
\nc{\rnc}{\renewcommand}
\nc{\lbar}[1]{\overline{#1}}
\nc{\bra}[1]{\langle#1|}
\nc{\ket}[1]{|#1\rangle}

\nc{\ketbra}[2]{|#1\rangle\!\langle#2|}
\nc{\dketbra}[2]{|#1\rangle\!\rangle\!\langle\!\langle#2|}

\nc{\braket}[2]{\langle#1|#2\rangle}

\nc{\kett}[1]{|#1\rangle\!\rangle}
\nc{\braa}[1]{\langle\!\langle#1|}

\nc{\proj}[1]{| #1\rangle\!\langle #1 |}
\nc{\avg}[1]{\langle#1\rangle}
\nc{\rank}{\operatorname{Rank}}
\nc{\smfrac}[2]{\mbox{$\frac{#1}{#2}$}}
\nc{\tr}{\operatorname{Tr}}
\nc{\ox}{\otimes}
\nc{\dg}{\dagger}
\nc{\dn}{\downarrow}
\nc{\cA}{{\cal A}}
\nc{\cB}{{\cal B}}
\nc{\cC}{{\cal C}}
\nc{\cD}{{\cal D}}
\nc{\cE}{{\cal E}}
\nc{\cF}{{\cal F}}
\nc{\cG}{{\cal G}}
\nc{\cH}{{\cal H}}
\nc{\cI}{{\cal I}}
\nc{\cJ}{{\cal J}}
\nc{\cK}{{\cal K}}
\nc{\cL}{{\cal L}}
\nc{\cM}{{\cal M}}
\nc{\cN}{{\cal N}}
\nc{\cO}{{\cal O}}
\nc{\cP}{{\cal P}}
\nc{\cQ}{{\cal Q}}
\nc{\cR}{{\cal R}}
\nc{\cS}{{\cal S}}
\nc{\cT}{{\cal T}}
\nc{\cU}{{\cal U}}
\nc{\cV}{{\cal V}}
\nc{\cX}{{\cal X}}
\nc{\cY}{{\cal Y}}
\nc{\cZ}{{\cal Z}}
\nc{\cW}{{\cal W}}
\nc{\csupp}{{\operatorname{csupp}}}
\nc{\qsupp}{{\operatorname{qsupp}}}
\nc{\var}{{\operatorname{var}}}
\nc{\rar}{\rightarrow}
\nc{\lrar}{\longrightarrow}

\nc{\RR}{{{\mathbb R}}}
\nc{\CC}{{{\mathbb C}}}
\nc{\FF}{{{\mathbb F}}}
\nc{\NN}{{{\mathbb N}}}
\nc{\ZZ}{{{\mathbb Z}}}
\nc{\PP}{{{\mathbb P}}}
\nc{\QQ}{{{\mathbb Q}}}
\nc{\UU}{{{\mathbb U}}}
\nc{\EE}{{{\mathbb E}}}
\nc{\id}{{\operatorname{id}}}

\nc{\CHSH}{{\operatorname{CHSH}}}

\nc{\<}{\langle}
\rnc{\>}{\rangle}
\nc{\rU}{\mbox{U}}

\nc{\Sym}{{\operatorname{Sym}}}

\newcommand{\CPTP}{\text{\rm CPTP}}

\usepackage{hyperref}
\hypersetup{colorlinks=true,citecolor=blue,linkcolor=blue,filecolor=blue,urlcolor=blue,breaklinks=true}

\makeatletter
\def\grd@save@target#1{%
  \def\grd@target{#1}}
\def\grd@save@start#1{%
  \def\grd@start{#1}}
\tikzset{
  grid with coordinates/.style={
    to path={%
      \pgfextra{%
        \edef\grd@@target{(\tikztotarget)}%
        \tikz@scan@one@point\grd@save@target\grd@@target\relax
        \edef\grd@@start{(\tikztostart)}%
        \tikz@scan@one@point\grd@save@start\grd@@start\relax
        \draw[minor help lines,magenta] (\tikztostart) grid (\tikztotarget);
        \draw[major help lines] (\tikztostart) grid (\tikztotarget);
        \grd@start
        \pgfmathsetmacro{\grd@xa}{\the\pgf@x/1cm}
        \pgfmathsetmacro{\grd@ya}{\the\pgf@y/1cm}
        \grd@target
        \pgfmathsetmacro{\grd@xb}{\the\pgf@x/1cm}
        \pgfmathsetmacro{\grd@yb}{\the\pgf@y/1cm}
        \pgfmathsetmacro{\grd@xc}{\grd@xa + \pgfkeysvalueof{/tikz/grid with coordinates/major step}}
        \pgfmathsetmacro{\grd@yc}{\grd@ya + \pgfkeysvalueof{/tikz/grid with coordinates/major step}}
        \foreach \x in {\grd@xa,\grd@xc,...,\grd@xb}
        \node[anchor=north] at (\x,\grd@ya) {\pgfmathprintnumber{\x}};
        \foreach \y in {\grd@ya,\grd@yc,...,\grd@yb}
        \node[anchor=east] at (\grd@xa,\y) {\pgfmathprintnumber{\y}};
      }
    }
  },
  minor help lines/.style={
    help lines,
    step=\pgfkeysvalueof{/tikz/grid with coordinates/minor step}
  },
  major help lines/.style={
    help lines,
    line width=\pgfkeysvalueof{/tikz/grid with coordinates/major line width},
    step=\pgfkeysvalueof{/tikz/grid with coordinates/major step}
  },
  grid with coordinates/.cd,
  minor step/.initial=.2,
  major step/.initial=1,
  major line width/.initial=2pt,
}
\makeatother

\usepackage{thmtools}
\usepackage{thm-restate}
\usepackage{etoolbox}

\makeatletter
\def\problem@s{}
\newcounter{problems@cnt}

\newcommand{\allproblems}{\problem@s}
\makeatother

\usepackage{array,mathtools,multirow,times,bm,tcolorbox,relsize,booktabs}
\usepackage[utf8]{inputenc}
\usepackage[T1]{fontenc}

\definecolor{colortwo}{rgb}{0.4,0.77,0.17}
\definecolor{colorthree}{rgb}{0.01,0.51,0.93}

\allowdisplaybreaks

\begin{document}
\title{Quantum Coherence and Distinguishability as Complementary Resources: \\A Resource-Theoretic Perspective from Wave-Particle Duality}

\author{Zhiping Liu}
\thanks{Zhiping Liu and Chengkai Zhu contributed equally to this work.} 
\affiliation{National Laboratory of Solid State Microstructures and School of Physics, Collaborative Innovation Center of Advanced Microstructures, Nanjing University, Nanjing 210093, China}
 \affiliation{Thrust of Artificial Intelligence, Information Hub,\\
The Hong Kong University of Science and Technology (Guangzhou), Guangzhou 511453, China}
\author{Chengkai Zhu} 
\thanks{Zhiping Liu and Chengkai Zhu contributed equally to this work.}
\affiliation{Thrust of Artificial Intelligence, Information Hub,\\
The Hong Kong University of Science and Technology (Guangzhou), Guangzhou 511453, China}
\author{Hua-Lei Yin}
\affiliation{Department of Physics and Beijing Key Laboratory of Opto-electronic Functional Materials and Micro-nano Devices, \\
Key Laboratory of Quantum State Construction and Manipulation (Ministry of Education), \\
Renmin University of China, Beijing 100872, China}
\affiliation{National Laboratory of Solid State Microstructures and School of Physics, Collaborative Innovation Center of Advanced Microstructures, Nanjing University, Nanjing 210093, China}
\author{Xin Wang}
\email{felixxinwang@hkust-gz.edu.cn}
\affiliation{Thrust of Artificial Intelligence, Information Hub,\\
The Hong Kong University of Science and Technology (Guangzhou), Guangzhou 511453, China}

\begin{abstract}
Wave-particle duality, a fundamental principle of quantum mechanics, encapsulates the complementary relationship between the wave and particle behaviors of quantum systems. In this paper, we treat quantum coherence and classical distinguishability as complementary resources and uncover a novel duality relation, which is explored through quantum state discrimination under incoherent operations, extending beyond typical interference scenarios. We prove that in an ensemble of mutually orthogonal pure states, the sum of `co-bits', quantifying the coherence preserved under incoherent free operations, and classical bits, representing the distinguishability extracted via quantum state discrimination, is bounded. This coherence-distinguishability duality relation exposes an inherent trade-off between the simultaneous preservation of a system's quantum coherence (wave-like property) and the extraction of its classical distinguishability (particle-like property). Our findings provide a fresh perspective on wave-particle duality through quantum resource theories, offering complementary insights into manipulating quantum and classical resources, with implications for quantum foundations and quantum technologies.

\end{abstract}

\date{\today}
\maketitle


\section{Introduction}

Niels Bohr introduced the complementarity principle as an essential feature of quantum mechanics~\cite{bohr1928quantum}, revealing a fundamental trade-off between two conjugate properties of a quantum system, such as position and momentum~\cite{busch1985indeterminacy}. This principle extends beyond a qualitative concept, connecting deeply with fundamental phenomena like the uncertainty principle~\cite{robertson1929uncertainty, maassen1988generalized, Coles_2014}, quantum nonlocality~\cite{Banik_2013}, and wave-particle duality. Complementarity also serves as a potential resource in practical applications, such as quantum cryptography~\cite{koashi2009simple, zhang2023quantum}, quantum walks~\cite{Kendon_2005}, and mutually unbiased bases~\cite{Petz_2007, Tavakoli_2021}.

Among the various manifestations of complementarity, wave-particle duality stands out as a striking example, exemplifying the continuous trade-off between wave-like and particle-like behaviors in quantum systems. As the wave behavior strengthens, the particle behavior weakens, and vice versa, highlighting its central role in exploring the essence of quantum theory~\cite{zhang2022building, Catani_2023, Catani_Lorenzo_2023,basso2021uncertainty}. This duality is typically demonstrated in interference experiments~\cite{menzel2012wave,wang2021molecular, yoon2021quantitative, chen2022experimental}. In double-slit setups, particle behavior is characterized by path information acquired by a which-path detector, while wave behavior is determined by the visibility of the interference pattern. Quantitative statements are formulated as an inequality~\cite{jaeger1995two, englert1996fringe}, namely wave-particle duality relation~\cite{wootters1979complementarity, greenberger1988simultaneous}:
\begin{equation}
\label{eq: WPDR}
        D^2+ V^2 \leq 1,
\end{equation}
where $D$ represents path distinguishability and $V$ denotes visibility of the interference fringe. 

\begin{figure}[t]
    \centering
    \includegraphics[width=1.03\linewidth]{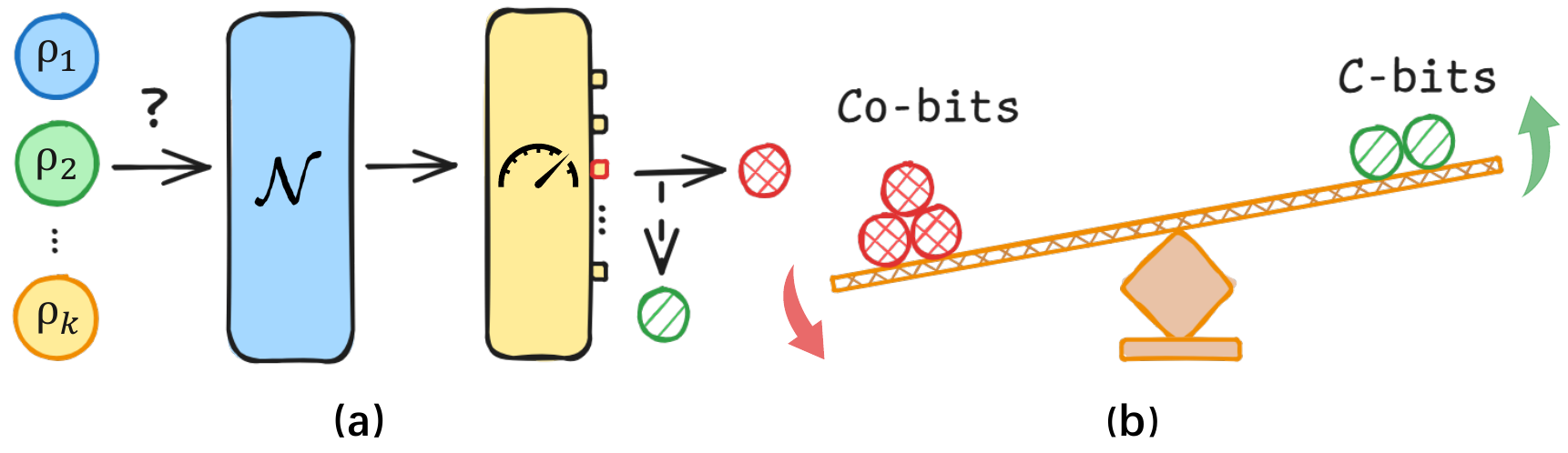}
    \caption{(a) Quantum state discrimination via free operations. The discrimination procedure consists of applying a free operation $\cN$ w.r.t. a quantum resource theory and measuring the state on a computational basis. 
    (b) Coherence-distinguishability duality relation. There is a trade-off between the quantum coherence resource (`co-bits') and the classical distinguishability resource (`c-bits').}
    \label{fig:duality} 
    \vspace{-0.3cm}
\end{figure}

Beyond double-slit setups, wave-particle duality has been extended to multipath interferometers~\cite{durr2001quantitative, bimonte2003comment} and reformulated within quantum information theory. These reformulations~\cite{Bera_2015,Bagan_2016, qureshi2017wave,Menon_2018, bagan2020wave, srivastava2021resource} leverage the resource theory of coherence~\cite{Winter_2016, Napoli_2016, Regula_2018,Fang_2018} and the fundamental quantum state discrimination (QSD) task~\cite{bae2015quantum}. In this context, coherence measures~\cite{Baumgratz_2014} replace visibility to capture wave properties effectively, while the distinguishability among detector states provides which-path information.  

In the resource-theoretic context, coherence is a fundamental quantum resource that generates superposition, which is essential for quantum algorithms~\cite{grover2000synthesis, Ahnefeld_2022} and intrinsic randomness, which is vital for quantum cryptography~\cite{Huttner_1995,Yuan_2015,ma2016quantum}. In contrast, distinguishability serves as a classical resource, enabling the extraction of deterministic classical information encoded in quantum systems through QSD~\cite{Jozsa_2000,hausladen1996classical,Terhal2001, Wang_2019_dis,salzmann2021symmetric}. 

From a broader perspective, independent of specific physical setups, understanding and characterizing the interplay between quantum resources and classical distinguishability has been a central focus of quantum information theory since its early days~\cite{Bennett1999b,Walgate2000,Bandyopadhyay2015,zhu2024limitations,zhu2025entanglement}. However, previous work has primarily focused on extending wave-particle duality into interference-based scenarios, such as many-particle systems~\cite{Dittel_2021, Menssen_2017, Brunner_2023} and discrimination games~\cite{Bagan_2018}, despite adapting a resource-theoretic perspective. A key question remains: What is the underlying relationship between quantum coherence and classical distinguishability in broader quantum information contexts, and does a quantitative complementarity exist beyond traditional interference experiments? Exploring this question could provide deeper physical insights into quantum resource theory and offer a fresh perspective on wave-particle duality. Moreover, it may offer a unified approach that enhances our understanding of the distinction between coherence and other quantum resources, such as entanglement and magic states, particularly in relation to their interaction with classical distinguishability.

In this work, we uncover a coherence-distinguishability duality relation within mutually orthogonal pure-state ensembles. Unlike previous works that assumed an interferometer scenario, our discovery stems from exploring coherence resource manipulation in discrimination tasks, an essential scenario offering operational characterization in general resource theory~\cite{Takagi_2019, takagi2019operational, Wu_2021,Regula_2021}, by introducing a framework for QSD via free operations (see Fig.~\ref{fig:duality}(a)). This duality relation shows that the more one extracts classical information from discriminating states, the less coherence can be preserved, and vice versa. It unveils an inherent trade-off between the coherence resource preserved after discrimination and the achievable perfect distinguishability within these ensembles (see Fig.~\ref{fig:duality}(b)). 

In particular, we present two intriguing cases. Firstly, we reveal that the duality relation is tight. When considering mutually orthogonal maximally coherent states, the sum of the maximum ``co-bits'' left and the classical bits extracted achieves the bound.  Secondly, when discriminating a complete orthonormal basis of a $d$-dimensional Hilbert space, no coherence can be preserved while extracting $\log_2 d$ classical bits. This case represents an extreme situation within our duality relation, highlighting the mutual exclusivity between coherence and distinguishability. Our work establishes a fundamental connection between coherence and distinguishability as complementary resources, extending the wave-particle duality relation into a new scenario within the realm of quantum resource theory.

\section{Quantum state discrimination via free operations}

We begin with an introduction to quantum resource theories and quantum state discrimination. Let $\cL(\cH_A, \cH_B)$ denote the set of linear operations from a $d_A$-dimensional Hilbert space $\cH_A$ to $\cH_B$. Let $\cD(\cH_A)$ be the set of density operators acting on $\cH_A$, and $\cN_{A\rightarrow B}$ be a quantum operation from system $A$ to $B$ which is a completely positive and trace-preserving map. A quantum resource theory of states is defined as a tuple $(\cF, \cO)$ where $\cF$ is the set of free states; $\cO$ is the set of free operations that preserve free states, i.e., $\cN(\rho) \in \cF,\, \forall \cN \in \cO, \forall \rho \in \cF$. For $(\cF, \cO)$, a reasonable resource measure $\cR(\rho) \in \mathbb{R}, \forall \rho \in \cD(\cH_A)$ satisfies monotonicity $\cR(\cN(\rho)) \leq \cR(\rho)$  and positivity $\cR(\rho) \geq 0, \forall \rho \in \cD(\cH_A), \forall \cN \in \cO$. Such a quantum resource theory intuitively arises when there is a restricted set of operations $\cO$ that are significantly easier to implement than the others, e.g., local operations and classical communication (LOCC) in entanglement theory~\cite{Chitambar_2014}, Clifford operations in the quantum resource theory of magic  states~\cite{Veitch_2014,Veitch2012}. 

For coherence, we fix the incoherent basis as the computational basis $\{\ket{i}\}_{i}$ and denote by $\cI$ the set of incoherent states, i.e., density matrices diagonal in $\{\ket{i}\}_{i}$. A $d$-dimensional maximally coherent state can be used to prepare any other state of the same dimension deterministically through incoherent free operations. One example is given by $\ket{\Psi_d}=\frac{1}{\sqrt{d}} \sum_{i=0}^{d-1}\ket{i}$, denoted as $\Psi_d = \ketbra{\Psi_d}{\Psi_d}$. The maximal set of operations that map incoherent states to incoherent states is called the maximally incoherent operations (MIO)~\cite{aberg2006quantifying}. The incoherent operations (IO)~\cite{Baumgratz_2014} require an incoherent Kraus decomposition $\{K_n\}_n$ such that $\frac{K_n\rho K_n^{\dagger}}{\tr[K_n\rho K_n^{\dagger}]} \in \mathcal{I}, \ \forall n, \rho \in \mathcal{I}$.
The dephasing-covariant incoherent operations (DIO)~\cite{chitambar2016critical, Marvian_2016} are those quantum operations $\mathcal{E}$ which commute with the dephasing operations $\Delta$ such that $[\Delta, \cE] = 0$. Finally, the strictly incoherent operations (SIO)~\cite{Winter_2016} are operations fulfilling both $\{K_n\}_n$ and $\{K^{\dagger}_n\}_n$ are sets of incoherent operators. There are hierarchies among these free operations: $\text{SIO} \subsetneq \textrm{IO} \subsetneq \textrm{MIO}$, $\text{SIO} \subsetneq \textrm{DIO} \subsetneq \textrm{MIO}$~\cite{chitambar2016critical}. More details about QRT of coherence can be found in Appendix~\ref{appendix: QRT_coherence}. 

Recall that for the minimum-error state discrimination, one aims to find a Positive Operator-Valued Measure (POVM) $\{E_j\}_{j=0}^{k-1}$ to maximize the average success probability of discriminating a state ensemble $\Omega = \{(p_j,\rho_j)\}_{j=0}^{k-1}$ with $\sum_j p_j =1, \rho_j\in \cD(\cH_A)$:
\begin{equation}
    P_{\rm suc}(\Omega) = \max_{\{E_j\}_j}\sum_{j} p_j\tr(E_j\rho_j).
\end{equation}
Numerous studies have been carried out to understand the resource quantification and limits of such a task when measurements are restricted to different classes, including POVMs with locality constraints~\cite{Childs2013,Bandyopadhyay2011a, chitambar2013local, chitambar2013revisiting} (i.e., LOCC, separable, PPT POVMs), incoherent~\cite{Oszmaniec_2019}, stabilizer measurement~\cite{zhu2024limitations}, all of which are considered within the context of different quantum resource theories.

To investigate the quantum resource manipulation in QSD, we reconsider a general framework for discriminating $\Omega = \{(p_j,\rho_j)\}_{j=0}^{k-1}$ as the following steps in order: 
(\uppercase\expandafter{\romannumeral1}) Receive an unknown state $\rho_j$ with prior probability $p_j$. (\uppercase\expandafter{\romannumeral2}) Apply a quantum channel $\cN_{A\rightarrow BA'}$ to $\rho_j$, yielding $\omega_{BA'}^{(j)} = \cN_{A\rightarrow BA'}(\rho_j)$ where $A'\cong A$ and $\dim\cH_B= k$. (\uppercase\expandafter{\romannumeral3}) Measure $\omega_{BA'}^{(j)}$ on the subsystem $B$ in computational basis. If the outcome is $i$, decide the received state is $\rho_i$.

We call the above quantum channel $\cN_{A\rightarrow BA'}$ a \textit{discrimination channel}. This framework provides an intuitive approach for assessing the `resourcefulness' of the discrimination process by restricting discrimination channels to free operations and incorporating the reference system $A'$, which helps track the degraded resources throughout the process. Specifically, for a given quantum state ensemble $\Omega = \{(p_j,\rho_j)\}_{j=0}^{k-1}$ and a quantum resource theory $(\mathcal{F},\mathcal{O})$, we introduce the optimal average success probability by free operations as follows.
\begin{equation*}\label{eq: free_QSD}
\begin{aligned}
    \widetilde{P}_{\rm suc, \cO}(\Omega) = &\;\sup \sum_{j=0}^{k-1} p_j \tr[\cN_{A\rightarrow BA'}(\rho_j)(\ketbra{j}{j}_B\ox I_{A'})]\\
    &\;{\rm s.t.}\;\; \cN_{A\rightarrow BA'} \in \cO,
\end{aligned}
\end{equation*}
where $\{\ket{j}\}_j$ is the computational basis.

Furthermore, we denote $\widetilde{P}_{\rm suc}(\Omega)$ as the maximal average success probability attainable, where $\cN_{A\rightarrow BA'}$ is optimized over all quantum channels. 
It is worth noting that $\widetilde{P}_{\rm suc,\cO}(\Omega)$ can be equivalently characterized by optimization over operations $\cN_{A\rightarrow B} \in \cO$, since discarding the subsystem $A'$ typically constitutes a free operation. The inclusion of the reference system $A'$, however, allows for a more complete quantification of the resources preserved during the discrimination process, as made clearer in Theorem~\ref{thm:roc_upper}. Note if there is an optimal $\cN_{A\rightarrow BA'} \in \cL(\cH_A, \cH_{BA'})$ such that $\widetilde{P}_{\rm suc}(\Omega) = P_{\rm suc}(\Omega)$, we say $\cN_{A\rightarrow BA'}$ can optimally discriminate $\Omega$.

\vspace{2mm}
\section{Coherence manipulation in state discrimination}
Now, in the context of quantum coherence manipulation in QSD, we present our first result that the process of optimally discriminating an ensemble $\Omega$ does not necessitate the consumption of additional quantum coherence. We establish this by demonstrating that the optimal average success probability, typically associated with unrestricted POVMs, can be realized through free operations within our discrimination framework.

\begin{proposition}
\label{prop:QSD_IO}
For a $d$-dimensional state ensemble $\Omega = \{(p_j,\rho_j)\}_{j=0}^{k-1}$,
the optimal discrimination probability $P_{\rm suc}(\Omega)$ can be achieved by incoherent operations. 
\end{proposition}

Proposition~\ref{prop:QSD_IO} is established by proving $P_{\rm suc}(\Omega) = \widetilde{P}_{\rm suc,IO}(\Omega) = \widetilde{P}_{\rm suc, MIO}(\Omega)$, with the detailed proof deferred to Appendix~\ref{appendix: coherence_qsd}. The wave-particle duality relation reminds us an intuition that increased visibility pattern always implies decreased path discrimination. This pivotal finding aligns with this intuition that extracted distinguishability cannot increase as coherence increases, indicating their complementary nature.

It is worth noting the relationship between quantum state discrimination via free operations and free POVMs. From the perspective of resource theories of measurements, the `free measurements' for coherence are previously considered as incoherent measurements~\cite{Oszmaniec_2019, Takagi_2019, takagi2019operational, Ducuara_2020}, which are diagonal in the fixed basis $\{\ket{i}\}_i$. In Appendix~\ref{appendix: coherence_qsd}, we show that the discrimination performed by incoherent measurements is equivalent to our aforementioned process with a discrimination channel belonging to DIO (SIO), a strict subset of MIO (IO). Combined with Proposition~\ref{prop:QSD_IO}, it reveals a strict hierarchy between MIO and DIO in terms of the optimal discrimination of coherent states. We can also see that perfect discrimination requires coherent POVMs but no coherence-generating processes.

When unrestricted POVMs alone cannot achieve optimal discrimination, one may introduce ancillary resource states by distinguishing $\Omega' =  \{(p_j,\rho_j \otimes \tau)\}_{j=0}^{k-1}$, where $\tau$ is some resource state. For example, the assistance of Bell states can enhance discrimination capabilities of PPT POVMs~\cite{Bandyopadhyay2009,Yu_2014,bandyopadhyay2015limitations, Bandyopadhyay_2021}. And a similar phenomenon appears in the resource theory of magic state~\cite{zhu2024limitations}. However, we note that even the maximally coherent state $\Psi_d$ is not helpful when using DIO (SIO), or equivalently, incoherent POVMs, for discrimination.

\begin{proposition}
\label{prop:QSD_cost}
For a $d$-dimensional state ensemble $\Omega = \{(p_j,\rho_j)\}_{j=0}^{k-1}$, the optimal discrimination probability via DIO cannot be improved with the assistance of any coherent state.
\end{proposition}

Proposition~\ref{prop:QSD_cost} implies a no-go case for reaching optimal discrimination in discriminating general coherent state ensembles via DIO (SIO), with the proof deferred to the appendix. When under more restricted incoherent operations, one can not extract more distinguishability even using more coherence resources.
Together with Proposition~\ref{prop:QSD_IO}, this result further emphasizes the unique complementarity between quantum coherence and classical distinguishability in resource manipulation, contrasting it with their interplay with resources like entanglement and non-stabilizerness.

\vspace{2mm}
\section{Coherence-distinguishability duality}
Given that optimal discrimination does not require additional coherence consumption, it is reasonable to postulate that the discrimination process may instead utilize the coherence inherent in the state ensemble. This prompts intriguing questions: how much coherence can be maximally preserved after a discrimination process? Is there a duality relation between the distinguishability and the maximally preserved coherence?

To address these questions, we start with the problem of discriminating an ensemble of mutually orthogonal pure states, which are foundational in quantum information theory~\cite{Bennett_1999, Ghosh_2001, Halder2019} and permit perfect discrimination. 
For coherence analysis, we consider the maximum relative entropy of coherence~\cite{Bu_2017}, denoted as $C_{\text{\rm max}}(\rho) = \log_2(1+C_{R}(\rho))$, where $C_{R}(\rho) = \min_{\sigma \in \cD(\cH_d)} \{s \geq 0 | \frac{\rho +s\sigma}{1+s} \coloneqq \tau \in \cI \}$ is the robustness of coherence. This measure is significant as it can be experimentally observed through a witness observable~\cite{Napoli_2016, zheng2018experimental}. The maximal value of $C_{\text{\rm max}}(\cdot)$ for a $d$-dimensional state is achieved by $\Psi_d$ with $C_{\text{\rm max}}(\Psi_d) = \log_2 d$~\cite{Piani_2016}, corresponding with $\log_2 d$ coherent bits (co-bits)~\cite{Chitambar_2016}. Then, we introduce the \textit{post-disrimination coherence} as the maximum average \textit{co-bits} that can be preserved after perfect discrimination.

\begin{definition}[Post-discrimination coherence]
\vspace{-0.2cm}
For a $d$-dimensional state ensemble $\Omega = \{(p_j,\rho_j)\}_{j=0}^{k-1}$, the post-discrimination coherence under maximally incoherent operations is defined as $\mathbf{C}_{\rm MIO}(\Omega):= \log_2 (1+ \eta)$, where
\begin{equation*}
\vspace{-0.3cm}
\begin{aligned}
    \eta :=\max_{\cN} &\; \sum_{j=0}^{k-1} p_j C_{R}(\sigma_j)\\
    \mathrm{s.t. }&\;\; \cN \in \mathrm{MIO},~\sigma_j = \tr_B[\cN_{A\rightarrow BA'}(\rho_j)],~\forall j,\\
    &\; \sum_{j=0}^{k-1} p_j \tr[\cN_{A\rightarrow BA'}(\rho_j)(\ketbra{j}{j}_B\ox I_{A'})] = P_{\rm suc}(\Omega).
\end{aligned}
\end{equation*}

\end{definition}

All constraints on $\cN$ ensure it is an MIO channel optimally discriminating $\Omega$, and $\{ \sigma_j \}_{j=0}^{k-1}$ act as degraded resource states that preserve coherence. This quantity evaluates the average resource retained in the quantum states after a discrimination process by maximally incoherent operations. It is worth noting that, by Proposition~\ref{prop:QSD_IO}, a feasible maximally incoherent operation $\cN$ always exists for $\mathbf{C}_{\rm MIO}(\Omega)$. Specifically, for a mutually orthogonal pure-state ensemble, the MIO channel $\cN$ can perfectly discriminate $\Omega$ and satisfies $\cN(\rho_j) = \ketbra{j}{j}\ox \sigma_j \, \forall j$. For such ensembles, we establish the connection between $\mathbf{C}_{\rm MIO}(\cdot)$ and $C_{\text{\rm max}}(\cdot)$ in Appendix~\ref{appendix: trade_off}, demonstrating the rationality of this definition. 

To characterize distinguishability, we recall that the distinguishability within a pure state ensemble $\Omega=\{(p_j,\ket{\psi_j})\}_{j=0}^{k-1}$ can be quantified by the von Neumann entropy of the ensemble's average state~\cite{Jozsa_2000}. Specifically, $\mathbf{S}(\Omega):=S(\hat{\omega})$ where $S(\cdot)$ is the von Neumann entropy of a state and $\hat{\omega} = \sum_j p_j \ketbra{\psi_j}{\psi_j}$ denotes the average state of $\Omega$. A mutually orthogonal pure-state ensemble $\Omega$ provides an ideal test bed for revealing fundamental quantum phenomena~\cite{yu2012four, halder2019strong, banik2021multicopy}. We now present our main result on the relationship between $\mathbf{S}(\Omega)$ and $\mathbf{C}_{\text{MIO}}(\Omega)$ for this ensemble.

\begin{theorem}[Coherence-distinguishability duality relation]\label{thm:roc_upper}
For a mutually orthogonal $d$-dimensional pure-state ensemble $\Omega = \{(1/k,\ket{\psi_j}) \}_{j=0}^{k-1}$,
\begin{equation}
\label{eq: roc_upper}  
    \mathbf{C}_{\text{MIO}}(\Omega) + \mathbf{S}(\Omega) \leq \log_2 d.
\end{equation}
\end{theorem}

Notice that from the state ensemble $\Omega$, one can distill $\mathbf{S}(\Omega)$ bits of classical information by perfectly discriminating any given unknown state within the ensemble. Theorem~\ref{thm:roc_upper} reveals a crucial trade-off: the more classical bits (`c-bits') you want to decode, the fewer `co-bits' can be preserved after extracting all classical information, and vice versa. This theorem is not entirely unexpected, as discrimination inherently introduces decoherence. However, it remarkably establishes a novel wave-particle duality relation akin to the form of Eq.~\eqref{eq: WPDR}, quantifying wave-like and particle-like behaviors in the discrimination process at the bit-level. The proof of Theorem~\ref{thm:roc_upper} is deferred to the appendix. 

More generally, we extend the coherence-distinguishability duality relation to the state ensemble with non-uniform distribution. By characterizing the distinguishability with the min-entropy, defined as $\mathbf{S}_{\min}(\Omega):= -H_{\max}(\hat{\omega}|| \mathbb{I})=
\mathbf{S}_{\min}(\hat{\omega})$, where $\hat{\omega}$ denotes the average state of $\Omega$ and $\mathbf{S}_{\min}(\hat{\omega}) = -\log_2 p_{\rm max}$ is the min-entropy of $\hat{\omega}$~\cite{Konig_2009}, then for a mutually orthogonal $d$-dimensional pure-state ensemble $\Omega = \{(p_j,\ket{\psi_j}) \}_{j=0}^{k-1}$,
\begin{equation}
    \mathbf{C}_{\text{MIO}}(\Omega) + \mathbf{S}_{\min}(\Omega) \leq \log_2 d.
\end{equation} 
This duality relation reveals that the sum of `co-bits' maximally preserved and `c-bits' at least gained is also bounded, although this bound is not tight for non-uniform state ensembles. The detailed proof and discussion are provided in Appendix~\ref{appendix: trade_off}.

\section{Boundary cases of duality}

To deepen our understanding of the coherence-distinguishability duality relation, we explore two specific instances illuminating its fundamental aspects. Firstly, analogous to how squeezed coherent states in quantum optics reach the Heisenberg uncertainty limit~\cite{walls1983squeezed}, we identify a particular state ensemble that saturates Eq.~\eqref{eq: roc_upper}. Let $H$ denote the $d$-dimensional Hadamard gate, given by $H = \frac{1}{\sqrt{d}}\sum_{i,j=0}^{d-1}\omega^{kj}\ketbra{k}{j}$ where $\omega = e^{i2\pi/d}$, and $X$ denote the $d$-dimensional generalized Pauli $X$ gate, given by $X = \sum_{i=0}^{d-1}\ketbra{i+1}{i}$. Then we have the following. 
\begin{proposition}\label{prop:max_exam}
Let $\ket{\phi_j} = HX^{j}\ket{0}$ where $H, X$ are the $d$-dimensional Hadamard gate and generalized Pauli $X$ gate. For $\Omega = \{(1/k,\ket{\phi_j})\}_{j=0}^{k-1}, \, k\leq d$,
\begin{equation}
\mathbf{C}_{\text{MIO}}(\Omega) + \mathbf{S}(\Omega) = \log_2 d.
\end{equation}
\end{proposition}

Proposition~\ref{prop:max_exam} shows that an ensemble of $k$ mutually orthogonal maximally coherent states exactly achieves the upper bound in Eq.~\eqref{eq: roc_upper}, which identifies the tightness of this trade-off relation. The proof is provided in Appendix~\ref{appendix: trade_off}. This finding underscores a novel role for maximally coherent states beyond their established status as golden resources within the quantum resource theory of coherence. It is also worth exploring other non-trivial ensemble cases saturating the coherence-distinguishability duality relation. A necessary condition for those ensembles to achieve the upper bound can be found in Appendix~\ref{appendix: trade_off}.

Furthermore, as another boundary case of Theorem~\ref{thm:roc_upper}, we note that when the cardinality of the set $\Omega$ is equal to the dimension $d$, $\mathbf{C}_{\text{MIO}}(\Omega)$ vanishes, indicating that no coherence could be preserved after perfect discrimination. Equivalently, this implies that if a quantum channel $\cN \in {\rm MIO}$ exists such that $\cN( \ketbra{\psi_j}{\psi_j}) = \ketbra{j}{j}\ox\sigma_j, \,\sigma_j\in\cH_d$ for each $j=0,1,...,d-1$, then each $\sigma_j$ must be an incoherent state. It illustrates an extreme case where it is inherently unfeasible to completely extract all c-bits encoded in a complete orthonormal basis of the Hilbert space through MIO while concurrently maintaining any co-bit. This scenario reveals a mutual exclusivity between coherence and distinguishability within the ensemble of mutually orthogonal pure states, analogous to the situation described by $D = 1$ and $V = 0$ in Eq.~\eqref{eq: WPDR}. Another extreme situation arises when the ensemble $\Omega$ contains only one state $\rho$, where no distinguishability can be obtained and the coherence of $\rho$ remains completely preserved, as $\mathbf{C}_{\text{MIO}}(\Omega) = C_{\text{\rm max}}(\rho)$. 
These observations hint again at an underlying wave-particle duality relationship within these ensembles. 

We note that previous work by Bera \textit{et al.}~\cite{Bera_2015} established a duality between coherence and path distinguishability via unambiguous state discrimination, extending the wave-particle duality to multipath interference. In their framework, \textit{particle states} and \textit{path detector states} are entangled within the interferometer, with coherence quantified by tracing out the detector states and path distinguishability derived from an ensemble of detector states for unambiguous discrimination. Orthogonal detector states indicate complete path distinguishability and the absence of coherence within the multipath interferometer.  
In interference scenarios, such as those explored in previous studies~\cite{Bera_2015, Bagan_2016}, non-orthogonal state ensembles explicitly demonstrate wave-particle duality with orthogonal ensembles serving as boundary cases. In contrast, our duality arises from arbitrary orthogonal pure-state ensembles. Similarly, in entanglement theory, orthogonal states, such as unextendible product bases, can still exhibit strong quantum nonlocality despite their mutual orthogonality~\cite{Bennett_1999}. This provides a broader, information-theoretical perspective beyond interference settings. 

\section{Conclusion}
In summary, we have explored the complementarity between quantum coherence and classical distinguishability
through resource manipulation in quantum state discrimination, revealing a coherence-distinguishability duality. Our duality highlights an inherent trade-off in the discrimination process: the more classical information (c-bits) one retrieves, the fewer coherence bits (co-bits) can be preserved, and vice versa. Notably, the coherence measure employed in our duality relation is directly experimentally observable, offering a clear physical setting that demonstrates wave-like behavior. Furthermore, with recent advances in state discrimination and coherence resource manipulation techniques in optical experiments~\cite{Sol_s_Prosser_2017, Wu_2020, Wu_2021_exp,sun2022activation, pan2012multiphoton}, our framework opens avenues for potential experimental realizations, offering deeper insights into this duality in practical contexts.

We anticipate that our framework for QSD could advance the understanding of the interplay between classical distinguishability and other quantum resources, such as entanglement, magic, and thermodynamics~\cite{Gour2015nonequilibrium, chiribella2022nonequilibrium, hsieh2025dynamical}. There is a possibility of exploring other duality relations between every two potential complementary resources.
Moreover, it would be interesting to investigate the potential generalization of this coherence-distinguishability duality to arbitrary state ensembles and explore whether an uncertainty relation can be formulated for this duality, considering that the wave-particle duality relation is a special case of the uncertainty relation. 

\vspace{6mm}
\section{Acknowledgments}
We thank Bartosz Regula for the suggestion of extending the duality relation into a pure state ensemble with non-uniform distribution and for other helpful discussions. We thank the anonymous
referees for their substantial and critical feedback, which has significantly enhanced the quality of our manuscript. We also thank Ranyiliu Chen, Kun Wang, and Benchi Zhao for their helpful comments. This work was partially supported by the National Key R\&D Program of China (Grant No.2024YFE0102500), the National Natural Science Foundation of China (Grant No. 12274223), the Guangdong Provincial Quantum Science Strategic Initiative (Grant No.~GDZX2403008, No.~GDZX2403001,
and No.~GDZX2303007), the Guangdong Provincial Key Lab of Integrated Communication, Sensing and Computation for Ubiquitous Internet of Things (Grant No.~2023B1212010007), the Quantum Science Center of Guangdong-Hong Kong-Macao Greater Bay Area, and the Education Bureau of Guangzhou Municipality.


\bibliography{ref}

\appendix

\onecolumngrid
\vspace{2cm}

\section{Quantum state discrimination via free operations}
\label{appendix: qsd_free}
Given a quantum state ensemble $\Omega = \{(p_j,\rho_j)\}_{j=0}^{k-1}$ where $\sum_j p_j =1, \rho_j\in \cD(\cH_A)$, $d_A = d$ and $k\leq d$, and a free states set $\cF$ with corresponding free operations set $\cO$, we introduce the \textit{ quantum state discrimination via channel}
\begin{enumerate}
    \item Receive an unknown state $\rho_j$ with prior probability $p_j$.
    \item Apply a quantum operation $\cN_{A\rightarrow A'B}$ to $\rho_j$ yielding $\tau_{A'B}^{(j)} = \cN_{A\rightarrow A'B}(\rho_j)$ where $A'\cong A$ and $\dim\cH_B= k$.
    \item Measure $\tau_{A'B}^{(i)}$ on subsystem $B$ in basis $\{\ket{i}\}_{i=0}^{k-1}$. If the outcome is $i$, decide the received state is $\rho_i$. 
\end{enumerate}

\begin{equation*}\label{Eq:pri_dual_SDP}
\begin{aligned}
&\underline{\textbf{POVM}}\\
P_{\rm suc}(\Omega) = \max_{\{E_j\}} & \;\; \sum_{j=0}^{k-1}p_j \tr(\rho_j E_j)\\
     {\rm s.t.}  & \;\; \sum_{j=0}^{k-1} E_j = I,E_j \geq 0, \, \forall j,
\end{aligned}
\quad\quad\quad
\begin{aligned}
&\underline{\textbf{Quantum channel}}\\
\widetilde{P}_{\rm suc}(\Omega) = &\max \sum_{j=0}^{k-1} p_j \tr[\cN_{A\rightarrow BA'}(\rho_j)(\ketbra{j}{j}_B\ox I_{A'})]\\
    {\rm s.t.} &\;\; \cN_{A\rightarrow BA'} \in\CPTP.
\end{aligned}
\end{equation*}

We call $\cN_{A\rightarrow BA'}$ the \textit{discrimination channel}, where subsystem $A'$ is introduced to further investigate the resource left after the discrimination process. When $\cN_{A\rightarrow BA'}$ is chosen from a set of free operations $\cO$ of some quantum resource theory, we call this task \textit{quantum state discrimination via free operations}. And we denote the optimal average success probability via free operations with $\widetilde{P}_{\rm suc, \cO}(\Omega)$. In the following, we discard the subsystem $A'$ and consider $\cN_{A \rightarrow B} \in \cO$ without affecting the attainment of $\widetilde{P}_{\rm suc, \cO}(\Omega)$.

Besides, we denote a set of restricted resourceless POVMs with $\cM_{\cF}$, when $\mathbf{E} = \{E_j \}_j  \in \cM_{\cF}$ for some resource theory, we call the task \textit{ quantum state discrimination via free POVMs} and denote the optimal average success probability via free POVMs with $P_{\rm suc, \cM_{\cF}}(\Omega)$.

\section{Quantum resource theory of coherence}
\label{appendix: QRT_coherence}
We briefly introduce the quantum resource theory (QRT) of coherence~\cite{streltsov2017colloquium, zhao2024probabilistic}. Let $\cH_d$ denote a Hilbert space of dimension $d$. Let $\cL(\cH_d)$ be the space of linear operators mapping $\cH_d$ to itself, and $\cD(\cH_d)$ be the set of density operators acting on $\cH_d$. The dephasing operations $\Delta$ are defined as follows.
\begin{equation}
\Delta(\rho) = \sum_{i=0}^{d-1}\ket{i}\bra{i}\rho\ket{i}\bra{i}
\end{equation}
The set of free states in the QRT of coherence is defined as $\mathcal{I} \coloneqq \{ \rho \geq 0| \Delta(\rho) = \rho \}$. In the following, we introduce several free operations in the QRT of coherence. The maximally incoherent operations (MIO)
is the largest class of incoherent operations, which map $\cI$ onto itself. The incoherent operations (IO)~\cite{Baumgratz_2014} admit a set of Kraus operators $\{K_n\}_n$ such that:
\begin{equation}
    \frac{K_n\rho K_n^{\dagger}}{\tr[K_n\rho K_n^{\dagger}]} \in \mathcal{I},\quad \sum_n K_n^{\dagger}K_n = \mathbb{I}, \ \forall n, \rho \in \mathcal{I}
\end{equation}
Dephasing-covariant incoherent operations (DIO)~\cite{chitambar2016critical, Marvian_2016} are those quantum operations $\mathcal{E}$ which commute with the dephasing operations $\Delta$ for any quantum state $\rho$ such that: $[\Delta, \cE] = 0$. The strictly incoherent operations (SIO)~\cite{Winter_2016} fulfill $\bra{i}K_n\rho K_n^{\dagger}\ket{i} = \bra{i}K_n \Delta[\rho] K_n^{\dagger}\ket{i} \ \forall n,i$. There are hierarchies among these free operations: $\text{SIO} \subsetneq \textrm{IO} \subsetneq \textrm{MIO}$, $\text{SIO} \subsetneq \textrm{DIO} \subsetneq \textrm{MIO}$~\cite{chitambar2016critical}.
 In a $d$-dimensional Hilbert space $\mathcal{H}_d$, a $d$-dimensional maximally coherent state is characterized by its ability to deterministically generate all other $d$-dimensional quantum states through incoherent operations~\cite{Baumgratz_2014}. A $d$-dimensional maximally coherent state is given by:   
\begin{equation}
\left|\Psi_d\right\rangle=\frac{1}{\sqrt{d}} \sum_{i=0}^{d-1}|i\rangle. 
\end{equation}
We denote $\ketbra{\Psi_d}{\Psi_d}$ with $\Psi_d$ in the following.

A good coherence measure is expected to satisfy the following three conditions under MIO:
\begin{itemize}
    \item $C(\rho) = 0 \ \forall \rho \in \cI$;
    \item $C(\rho) \geq C(\cT(\rho))$ for all incoherent CPTP maps $\cT$;
    \item Convexity: $\sum_j p_j C(\rho_j) \geq C(\sum_j p_j \rho_j)$, which is not necessary. 
\end{itemize}

\begin{definition}(Robustness of coherence~\cite{Napoli_2016, D_az_2018})
The robustness of coherence of a quantum state $\rho \in \cD(\cH_d)$ is defined as:
\begin{align}
    C_{R}(\rho) = \min_{\tau \in \cD(\cH_d)} \left\{s \geq 0 \;\bigg|\; \frac{\rho +s\tau}{1+s} \coloneqq \delta \in \cI  \right\},
\end{align}
\end{definition}
which can be transformed into a simple semidefinite program (SDP)~\cite{Napoli_2016}:
\begin{align}
    \label{def: ROC_SDP}
    C_{R}(\rho) = \max &\; \tr(W\rho) \\
     {\rm s.t.} & \;\; \Delta(W) \leq 0,\\
     & \;\; W \geq -I,
\end{align} 
where the Hermitian operator $M = -W$ fulfills $M \geq 0$ if and only if $\tr(\rho M) = \tr(\rho \Delta(M))$ for all incoherent states $\rho \in \cI$. Such an observable $M$ serves as a coherence witness, where $\tr(\rho M) \leq 0$ indicates coherence in the state $\rho$. Robustness of coherence is multiplicative under the tensor product of states:
\begin{equation}
    \label{eq: ROC_multip}
    1 + C_{R} (\rho_1 \ox \rho_2) = ( 1 + C_{R} (\rho_1))( 1 + C_{R} (\rho_2))
\end{equation}
Another main advantage of the robustness of coherence is that it can be estimated in the laboratory as the expected value of observable $M$ with respect to $\rho$. An operational interpretation of the robustness of coherence is that: it quantifies the advantage enabled by a quantum state in a \textit{phase discrimination task}.

Another equivalent primal standard form of the above SDP~\cite{D_az_2018} is:
\begin{align}
   1 + C_{R}(\rho) = \min \big\{\lambda\, |\, \rho \leq \lambda \sigma, \sigma \in \cI\big\}
\end{align}
and the dual form is given by 
\begin{align}\label{sdp: ROC_e_dual}
   1 + C_{R}(\rho) = \max \big\{\tr(\rho S) \, |\, S \geq 0, S_{ii} = 1, \ \forall i\big\}
\end{align}

\begin{definition}(Maximum relative entropy of coherence)~\cite{Bu_2017}\label{def: entropy_co}
The maximum relative entropy of coherence of a state $\rho$ is defined as:
\begin{equation}
    C_{\max}(\rho) = \min_{\sigma \in \cI} D_{\max}(\rho || \sigma),
\end{equation}
where $\cI$ is the set of incoherent states in $\cD(\cH_d)$ and $D_{\max}(\rho || \sigma)$ denotes the maximum relative entropy of $\rho$ with respect to $\sigma$ and $D_{\max}(\rho || \sigma) \coloneqq \min \{ \lambda\, |\, \rho \leq 2^{\lambda}\sigma \}$.
\end{definition}
It is obvious that $D_{\max}(\rho || \sigma)$ is the upper bound of the maximum relative entropy of coherence. Note that $2^{C_{\max}(\rho)} = 1 + C_{R}(\rho)$.

\section{Coherence manipulation within QSD}
\label{appendix: coherence_qsd}

\renewcommand{\theproposition}{S\arabic{proposition}}
 \begin{lemma}
    \label{lem: QSD_QC}
    For a $d$-dimensional state ensemble $\Omega = \{(p_j,\rho_j)\}_{j=0}^{k-1}$ and a resource theory of states $(\cF, \cO)$, there exist 
    quantum-classical channels $\cN^{\text{qc}} \in \cO$ achieving the optimal discrimination probability via free operations $\widetilde{P}_{\rm suc, \cO}(\Omega)$. 
\end{lemma}

\begin{proof}
    We consider QSD via free operations in some quantum resource theory with $\cN \in \cO$.
    Suppose $\cN$ is the optimal channel achieving $\widetilde{P}_{\rm suc, \cO}(\Omega)$. It follows
    \begin{subequations}
        \begin{align}
         \label{eq: MED_channel_Dephase}
            \widetilde{P}_{\rm suc,\cO}(\Omega) &= \sum_{j=0}^{k-1} p_j \tr(\cN(\rho_j) \ketbra{j}{j}) \\
            & = \sum_{j=0}^{k-1} p_j \tr \big(\cN(\rho_j)\Delta (\ketbra{j}{j}) \big) \\
            &= \sum_{j=0}^{k-1} p_j \tr(\Delta\circ\cN(\rho_j)\ketbra{j}{j}),
        \end{align}
    \end{subequations}
where we use the fact that the adjoint map of a fully dephasing channel $\Delta(\cdot)$ is itself. Note that $\cN \in \cO$ is a free operation in some resource theory; it holds that $\cN(\rho) \in \cF ,\ \forall \rho \in \cF$, where $\cF$ denotes the set of free states. Obviously, we have $\Delta \circ \cN(\rho) \in \cF ,\ \forall \rho \in \cF$. Thus, it holds $\Delta \circ \cN  \in \cO$. Then we conclude that $\Delta \circ \cN$ is also an optimal free channel achieving $\widetilde{P}_{\rm suc, \cO}(\Omega)$. Then we show that $\Delta \circ \cN$ is a quantum-classical channel.
We express $J_{\cN} = \sum_{i,i'} \ketbra{i}{i'} \ox \cN(\ketbra{i}{i'})$ and deduce that  
\begin{subequations}
    \begin{align}
        J_{\Delta \circ \cN} &= \sum_{i,i'} \ketbra{i}{i'} \ox \Delta \circ \cN(\ketbra{i}{i'}) \\
        & = \sum_{i,i', q} p_{i,i', q} \ketbra{i}{i'} \ox \ketbra{q}{q}) \\
        & = \sum_{q} Q_q \ox \ketbra{q}{q}
    \end{align}
\end{subequations}
Then we express the Choi operator of $\Delta \circ \cN$ as $J_{\Delta \circ \cN} = \sum_{q} Q_q \ox \ketbra{q}{q}$, where $Q_q = \sum_{i,i'} p_{i,i',q} \ketbra{i}{i'}$, 
$\sum_{q} Q_q = I$ and $Q_q \geq 0$ due to $\Delta \circ \cN$ is a CPTP map. We can choose $Q_q = M_q^T$ and conclude $\{ M_q \}_{q = 0}^{k-1}$ is a POVM. Thus, $\Delta \circ \cN$ is a quantum-classical channel $\cN^{\text{qc}}$. Such conclusion also holds when $\cN \in \CPTP$ without considering any resource theory, and indicates that $\widetilde{P}_{\rm suc}(\Omega) \leq P_{\rm suc}(\Omega)$. 
\end{proof}

\renewcommand\theproposition{\textcolor{blue}{1}}
\setcounter{proposition}{\arabic{proposition}-1}

\begin{proposition}

For a $d$-dimensional state ensemble $\Omega = \{(p_j,\rho_j)\}_{j=0}^{k-1}$,
the optimal discrimination probability $P_{\rm suc}(\Omega)$ can be achieved by quantum state discrimination via incoherent operations. 
\end{proposition}

\begin{proof}
We prove this proposition by demonstrating that $P_{\rm suc}(\Omega) = \widetilde{P}_{\rm suc,IO}(\Omega)$. First, we will show $P_{\rm suc}(\Omega) \leq \widetilde{P}_{\rm suc,IO}(\Omega)$. Suppose the optimal POVM for $P_{\rm suc}(\Omega)$ is $\{E_j\}_{j=0}^{k-1}$. We can obtain a quantum-classical channel
\begin{equation}
    \cM_{A\rightarrow B}(\rho_A) = \sum_{j=0}^{k-1}\tr(\rho_A E_j)\ketbra{j}{j}.
\end{equation}
In the following, we show that quantum-classical channel $\cM_{A\rightarrow B}$ is an IO.
Suppose the spectral decomposition of $E_q$ is 
\begin{equation}
    E_q = \sum_{i = 1}^{r_q} \lambda_i^q \ketbra{\psi_i^q}{\psi_i^q},
\end{equation}
where $r_q$ is the rank of $E_q$. Then we can express the Kraus operators of the quantum-classical channel as $K_i^q = \sqrt{\lambda_i^q} \ketbra{q}{\psi_i^q}$.  We can check that for any $\ket{j}$:
\begin{align}
    K_i^q \ketbra{j}{j} (K_i^q)^{\dagger} = \lambda_i^q \ket{q}\langle\psi_i^q|j\rangle\langle{j}|\psi_i^q\rangle\langle q| = \lambda_i^q  |\langle j|\psi_i^q\rangle|^2 \ketbra{q}{q}.
\end{align}
$K_i^q \ketbra{j}{j} (K_i^q)^{\dagger} \in \cI$. Thus, we can conclude that $\cM_{A\rightarrow A'} \in \text{IO}$. 
Then we have
\begin{equation}
\widetilde{P}_{\rm suc, IO}(\Omega) \geq \sum_{j=0}^{k-1} p_j \tr(\cM(\rho_j)\ketbra{j}{j}) = \sum_{j=0}^{k-1} p_j \tr(\rho_j E_j) = P_{\rm suc}(\Omega).
\end{equation}
Second, we are going to show that $P_{\rm suc}(\Omega) \geq \widetilde{P}_{\rm suc, IO}(\Omega)$. Recall that $\widetilde{P}_{\rm suc, IO}(\Omega)$ can always be achieved by a quantum-classical channel $\cN'$ according to Lemma~\ref{lem: QSD_QC}. Note that we have shown that the quantum-classical channel is an IO. Suppose that $J_{N'} = \sum_{q} M_q^T \ox \ketbra{q}{q}$ and we can conclude that $\widetilde{P}_{\rm suc, IO}(\Omega)$ is equivalently achieved by a POVM $\{ M_q \}_{0}^{k-1}$. Thus, $P_{\rm suc}(\Omega) \geq \widetilde{P}_{\rm suc, IO}(\Omega)$. 

In conclusion, we have $P_{\rm suc}(\Omega) = \widetilde{P}_{\rm suc, IO}(\Omega)$, which means optimal discrimination probability via free operations can be achieved with $\text{IO}$. The result also holds for $\text{MIO}$ because $\text{IO} \subsetneq \text{MIO}$.
\end{proof}

\begin{definition} (Incoherent measurement~\cite{Oszmaniec_2019})
A $d$-dimentional POVM $\{ E_m \}_{m=0}^{k-1}$ is called an incoherent measurement if $\Delta(E_m) = E_m$ for all $m$.
\end{definition}
Incoherent measurements can be regarded as POVMs' analog of incoherent states, which are diagonal in the fixed basis and can be expressed as $E_m = \sum_{i = 0}^{d-1} p(m|i) \ketbra{i}{i}$. We denote the set of incoherent measurements as $\cM_{\cI}$.

\renewcommand{\theproposition}{S\arabic{proposition}}
 \begin{lemma}
    \label{lem: POVM_DIO_Channel}
    For a $d$-dimensional state ensemble $\Omega = \{(p_j,\rho_j)\}_{j=0}^{k-1}$, we have
    \begin{equation}
        P_{\rm suc, \cM_{\cI}}(\Omega) = \widetilde{P}_{\rm suc,SIO}(\Omega).
    \end{equation}
 \end{lemma}

\begin{proof}
First, we show that $P_{\rm suc, \cM_{\cI}}(\Omega) \geq \widetilde{P}_{\rm suc,DIO}(\Omega)$. Recall that the discrimination channel $\cN \in \text{DIO}$ reaching $\widetilde{P}_{\rm suc,DIO}(\Omega)$ can be a quantum-classical channel. We denote the quantum-classical channel $\cN$ by $\cN (\cdot) = \sum_{j=0}^{k-1} \tr(E_j\cdot)\ketbra{j}{j}$, where $\{E_j\}_{j=0}^{k-1}$ is the corresponding POVM. Note that $\cN \in \text{DIO}$ and $\forall \rho$, we have 
\begin{equation}
    \cN(\rho)= \Delta \circ \cN(\rho) = \cN \circ \Delta (\rho) = \sum_{j=0}^{k-1} \tr(E_j\Delta(\rho))\ketbra{j}{j} = \sum_{j=0}^{k-1} \tr(\Delta(E_j)\rho)\ketbra{j}{j}.
\end{equation}
Obviously, $\{\Delta(E_j)\}_{j=0}^{k-1}$ is an incoherent measurement with $\sum_{j=0}^{k-1} \Delta(E_j) = I$ and $\Delta(E_j) \geq 0$. Thus, we can directly let $E_j = \sum_{i} p_{i, j} \ketbra{i}{i}$ and $\{E_j\}_{j=0}^{k-1} \in \cM_{\cI}$. Thus, we conclude that $P_{\rm suc, \cM_{\cI}}(\Omega) \geq \widetilde{P}_{\rm suc,DIO}(\Omega)$. 

Conversely, if we distinguish $\Omega$ via incoherent measurements, and the optimal discrimination probability under such constraints is achieved by an incoherent measurement $\{E_j\}_{j=0}^{d-1}$, where $E_j = \sum_{i} p_{i, j} \ketbra{i}{i}$. We can construct a channel $\widetilde{\cN} \in \text{DIO}$ with its Choi operator being $J_{\widetilde{\cN}} =  =\sum_{j} E_j \ox \ketbra{j}{j} = 
\sum_{i,j} p_{i, j} \ketbra{i}{i} \ox \ketbra{j}{j} $, which means $P_{\rm suc, \cM_{\cI}}(\Omega) \leq \widetilde{P}_{\rm suc,DIO}(\Omega)$.
We can further reduce the DIO to SIO. We have shown that the Choi operator of the quantum-classical channel of incoherent measurements can be expressed as $J_{\widetilde{\cN}} =  \sum_{i,j} p_{i, j} \ketbra{i}{i} \ox \ketbra{j}{j}$ with $E_j = \sum_i p_{i,j}\ketbra{i}{i}$. Then we can express the Kraus operators of $J_{\widetilde{\cN}}$ as $K_i^j = \sqrt{p_{i,j}} \ketbra{j}{i}$. We can check that 
\begin{subequations}
    \begin{align}
        K_i^j\ketbra{q}{q'}(K_i^j)^{\dagger}= p_{i,j} \ketbra{j}{i}q\rangle \langle q' \ketbra{i}{j}= p_{i,j} \delta_{i,q} \delta_{i,q'} \ketbra{j}{j}
    \end{align}
\end{subequations}
If $q = q'$, we obtain that $K_i^j\ketbra{q}{q}(K_i^j)^{\dagger} = p_{i,j} \delta_{i,q}\ketbra{j}{j}$; and if $q \neq q'$, we obtain that $K_i^j\ketbra{q}{q'}(K_i^j)^{\dagger} = 0$. Thus $\widetilde{\cN}$ discussed above belongs to $\text{SIO}$. Thus, we have $ P_{\rm suc, \cM_{\cI}}(\Omega) = \widetilde{P}_{\rm suc,SIO}(\Omega)$. It also holds for $ P_{\rm suc, \cM_{\cI}}(\Omega) = \widetilde{P}_{\rm suc,DIO}(\Omega)$.
\end{proof}

\renewcommand\theproposition{\textcolor{blue}{2}}
\setcounter{proposition}{\arabic{proposition}-1}
\begin{proposition}
    \label{Prop: POVM_DIO_Cost}
    For a $d$-dimensional state ensemble $\Omega = \{(p_j,\rho_j)\}_{j=0}^{k-1}$, the optimal discrimination probability via DIO $\widetilde{P}_{\rm suc,DIO}(\Omega)$  cannot be improved with the assistance of any coherent state.
\end{proposition}

\begin{proof}
    Recall the proof of Lemma~\ref{lem: POVM_DIO_Channel}. Suppose that the optimal channel $\cN \in {\rm DIO}$ achieving $\widetilde{P}_{\rm suc, DIO}(\Omega)$ can be donated as $\cN (\cdot) = \sum_{j=0}^{k-1} \tr(E_j\cdot)\ketbra{j}{j}$, where $\{E_j\}_{j=0}^{k-1}$ is the corresponding incoherent POVM. Now, consider the input state $\rho$ with an ancillary coherent coherent state $\sigma$,
    we have 
    \begin{subequations}
    \begin{align}
        \cN(\rho \otimes \sigma) &= \Delta \circ \cN(\rho \otimes \sigma)\\
        & = \cN \circ \Delta (\rho \otimes \sigma)\\
        &= \sum_{j=0}^{k-1} \tr\Big(E_j\Delta(\rho \otimes \sigma)\Big)\ketbra{j}{j} \\
        &= \sum_{j=0}^{k-1} \tr\Big(E_j\Delta(\rho) \otimes \Delta (\sigma)\Big)\ketbra{j}{j} \\
        & =  \sum_{j=0}^{k-1} \tr\Big(E'_j\Delta(\rho)\Big)\ketbra{j}{j},
    \end{align}
\end{subequations}
where $\{E'_j\}_{j=0}^{k-1}$ is a new incoherent POVM, and $\cN(\cdot \otimes \sigma) $ remains a DIO channel. Thus, the assistance of coherent states can not improve $\widetilde{P}_{\rm suc,DIO}(\Omega)$. Furthermore, combined with Lemma~\ref{lem: POVM_DIO_Channel}, we can also conclude that the optimal discrimination probability via incoherent POVMs $P_{\rm suc, \cM_{\cI}}(\Omega)$ also cannot be improved with the assistance of coherent states.
\end{proof}

\begin{remark}
Until now, we have demonstrated the significant difference between MIO and DIO in QSD tasks. Using the paradigm of QSD via free operations, we can implement the optimal POVMs equivalently through MIO (IO), but only achieve the incoherent measurements equivalently through DIO (SIO). Thus, we can argue that MIO (IO) provides a maximal advantage over DIO (SIO) in the QSD task, as discussed in Ref.~\cite{Oszmaniec_2019, Takagi_2019}. This advantage is fully identified by the quantifier called robustness for incoherent measurements $\cR_{\cM_{\cI}}$ as follows:
\begin{equation}
    \cR_{\cM_{\cI}}(\mathbf{M}) = \max_{\Omega_0} \frac{p_{\mathrm{suc}}(\Omega_0, \mathbf{M})}{\max_{\mathbf{N} \in \cM_{\cI}} p_{\mathrm{suc}} (\Omega_0, \mathbf{N})} - 1,
\end{equation}
where $p_{\mathrm{suc}}(\Omega_0, \mathbf{E}) = \sum_j p_j \tr(E_j\rho_j)$ with $\mathbf{E} = \{E_j \}_{j=0}^{k-1} $, $\mathbf{M}$ is a unrestricted POVM, $\mathbf{N}$ is an incoherent measurement, and $\Omega_0 = \{p_j, \rho_j \}_{j=0}^{k-1}$ denotes a quantum state ensemble. Note that
$\cR_{\cM_{\cI}}(\mathbf{M}) \geq 0$ and when $\mathbf{M} \nsubseteq \cM_{\cI}$, $\exists \ \Omega_0$ let $\cR_{\cM_{\cI}}(\mathbf{M}) > 0$. We can conclude 
that QSD via free operation in the QRT of coherence can identify a strict hierarchy between MIO and DIO on the optimal discrimination probability.
\end{remark}

\section{Coherence-distinguishability duality relation}
\label{appendix: trade_off}

\renewcommand{\theproposition}{S\arabic{proposition}}
 \begin{lemma}
    \label{lem: roc_ensemble}
    For a mutually orthogonal $d$-dimensional state ensemble $\Omega = \{(p_j,\rho_j)\}_{j=0}^{k-1}$, if the post-discrimination coherence under maximally incoherent operations $\mathbf{C}_{\rm MIO}(\Omega)$ is achieved by some MIO $\cN$, then $\mathbf{C}_{\rm MIO}(\Omega) = C_{\text{\rm max}}(\cN(\hat{\omega}))$ with $\hat{\omega} = \sum_{j=0}^{k-1} p_j \rho_j$.
   
\end{lemma}

\begin{proof}
According to the definition, the MIO $\cN$ achieving $\mathbf{C}_{\rm MIO}(\Omega)$ satisfies $\cN(\rho_j) = \ketbra{j}{j}\ox \sigma_j, \forall \rho_j \in \Omega$. Then we have 
\begin{subequations}
\begin{align}
   \mathbf{C}_{\rm MIO}(\Omega) & = \log_2 [1+ \eta] \\
   & =  \log_2 [1+ \sum_{j=0}^{k-1} p_j C_{R}(\sigma_j)]\label{eq0} \\
   & =  \log_2 [1+ \sum_{j=0}^{k-1} p_j C_{R}( \ketbra{j}{j}\ox \sigma_j)] \label{eq1}\\
   & =  \log_2 \left[1+  C_{R}\Big(\sum_{j=0}^{k-1} p_j \ketbra{j}{j}\ox \sigma_j\Big)\right] \label{eq2}\\
   & =  \log_2 \left[1+  C_{R}\Big(\sum_{j=0}^{k-1} p_j \cN(\rho_j)\Big)\right] \\
    & =  \log_2 [1+  C_{R}\Big(\cN(\hat{\omega})\Big)]\\
    & =  C_{\text{\rm max}}(\cN(\hat{\omega})) ,
\end{align}
\end{subequations}
where $\hat{\omega} = \sum_{j=0}^{k-1} p_j \rho_j$. 
We drive Eq.~\eqref{eq1} from Eq.~\eqref{eq0} because the robustness of coherence is multiplicative under the tensor product of states. And Eq.~\eqref{eq2} is driven from Eq.~\eqref{eq1} because the robustness of coherence is convex linear on classical-quantum states~\cite{gour2024resources}.
\end{proof}

In this Lemma~\ref{lem: roc_ensemble}, we reveal the connection between post-discrimination coherence under maximally incoherent operations, $\mathbf{C}_{\rm MIO}(\cdot)$, and the maximum relative entropy of coherence, $C_{\text{\rm max}}(\cdot)$, thereby demonstrating the rationality of this definition.

\renewcommand{\theproposition}{S\arabic{proposition}}

\begin{theorem}
\label{theo:dual_min}
For a mutually orthogonal $d$-dimensional pure-state ensemble $\Omega = \{(p_i,\ket{\psi_i}) \}_{i=0}^{k-1}$,
\begin{equation}
    \mathbf{C}_{\text{MIO}}(\Omega) + \mathbf{S}_{\rm min}(\Omega) \leq \log_2 d.
\end{equation}
\end{theorem}

\begin{proof}
Denote by $\cE = \{ \ket{\psi_i} \}_{i=0}^{d-1}$ an orthogonal basis of $\cH_A$ with $d_A = d$. Let $\cN$ be the optimal channel achieving $\mathbf{C}_{\text{MIO}}(\Omega)$ with $\cN(\ketbra{\psi_i}{\psi_i} ) = \ketbra{i}{i}\ox\sigma_i$, where $ \ketbra{i}{i} \in \cD(\cH_B)$, $d_B = k$ and $ i=0, \cdots, k-1$.
For these states $\ket{\psi_j}, j = k, \cdots, d-1$, i.e., $\ket{\psi_j} \in \cE$ but $\ket{\psi_j} \notin \Omega$, we directly assume:
\begin{equation}\label{eq: out_ensemble}
    \cN(\ketbra{\psi_j}{\psi_j}) = \left(\sum_{n=0}^{k-1} q_n^j\ketbra{n}{n}\right) \ox \rho_j, \quad j=k,\cdots d-1, 
\end{equation}
where $\sum_{n=0}^{k-1} q_n^j = 1$, $\rho_j \in \cD(\cH_{A'})$ and $d_{A'} = d$,  and we use the conclusion that $(\Delta \ox I) \circ \cN$ is also the optimal channel achieving $\mathbf{C}_{\text{MIO}}(\Omega)$.
Then we have,
\begin{equation}
\begin{aligned}
    \cN\left(\frac{I}{d}\right) &= \cN\left(\frac{1}{d}\sum_{i=0}^{d-1}\ketbra{\psi_i}{\psi_i}\right) \\
    &=\frac{1}{d} \left(\sum_{i=0}^{k-1}\cN(\ketbra{\psi_i}{\psi_i}) + \sum_{j=k}^{d-1}\cN(\ketbra{\psi_j}{\psi_j}) \right) \\
     &= \frac{1}{d} \left(\sum_{i=0}^{k-1} \ketbra{i}{i}\ox \sigma_i + \sum_{j=k}^{d-1} \left(\sum_{n=0}^{k-1} q_n^j\ketbra{n}{n}\right) \ox \rho_j\right) \\
     & = \frac{1}{d} \left(\sum_{i=0}^{k-1} \ketbra{i}{i}\ox \left(\sigma_i + \sum_{j=k}^{d-1} q_i^j \rho_j\right)\right) 
\end{aligned}
\end{equation}
Notice that $\cN$ is an MIO which yields $\cN\left(\frac{I}{d}\right)$ is incoherent. Thus, we have $\sigma_i + \sum_{j=k}^{d-1} q_i^j \rho_j$ is an incoherent state (unnormalized) for $i = 0,\cdots, k-1$. We can conclude that $C_{R}(\sigma_i) \leq \sum_{j=k}^{d-1} q_i^j$. Therefore, 
\begin{equation}
\label{eq: upper_ROC}
    \eta =  \sum_{i=0}^{k-1}p_iC_{R}(\sigma_i)  \leq \max \sum_{i=0}^{k-1}p_i\sum_{j=k}^{d-1} q_i^j.
\end{equation}
If each $\ket{\psi_i}$ in $\Omega$ is given with probability $p_i$, we have 
\begin{equation}
\label{eq:bound}
    \sum_{i=0}^{k-1}\sum_{j=k}^{d-1} p_i q_i^j \overset{\text{(i)}}{\leq} \sum_{j=k}^{d-1}\sum_{i=0}^{k-1}p_{\text{\rm max}}q_i^j = p_{\text{\rm max}} \cdot\sum_{j=k}^{d-1}\sum_{i=0}^{k-1}q_i^j = p_{\text{\rm max}} \cdot (d-k),
\end{equation}
where $p_{\text{\rm max}} = \max (\{p_0,p_1,...,p_{k-1}\})$. Thus, we have $\mathbf{C}_{\text{MIO}}(\Omega) = \log_2(1 + \eta) \overset{\text{(ii)}}{\leq}  \log_2[p_{\text{\rm max}}\cdot (d-k)+1] \overset{\text{(iii)}}{\leq} \log_2 d + \log_2 (p_\text{\rm max})$ and conclude
\begin{equation}
\label{eq:coherence_dis_bound}
    \mathbf{C}_{\text{MIO}}(\Omega) + \mathbf{S}_{\rm min}(\Omega)\leq \log_2 d,
\end{equation}
where $\mathbf{S}_{\rm min}(\Omega)$ is defined as the min-entropy of the average state of $\Omega$ with $\hat{\omega} = \sum_i p_i\ketbra{\psi_i}{\psi_i}$ and $\mathbf{S}_{\rm min}(\Omega) = \mathbf{S}_{\rm min}(\hat{\omega} ) = - \log_2 (p_\text{\rm max})$. 
\end{proof}

In Theorem~\ref{theo:dual_min}, we derive a bound between coherence and distinguishability, where $\mathbf{S}_{\rm min}$ can be interpreted as the `c-bits' at least gained through perfect discrimination. We briefly discuss the tightness of this bound in terms of whether the inequality can be saturated. When the state ensemble $\Omega$ has a non-uniform distribution, condition (i) in Eq.~\eqref{eq:bound} is strictly less than, leading to condition (ii) also being strictly less than. As a result, we have $\mathbf{C}_{\text{MIO}}(\Omega) + \mathbf{S}_{\rm min}(\Omega) < \log_2 d$, indicating that the bound is not tight when $\Omega$ is non-uniform. Thus, a uniform ensemble is necessary for saturating Eq.~\eqref{eq:coherence_dis_bound}.

\renewcommand\theproposition{\textcolor{blue}{2}}
\setcounter{proposition}{\arabic{proposition}-1}

\begin{theorem}
\label{thm:roc_upper}
For a mutually orthogonal $d$-dimensional pure-state ensemble $\Omega = \{(1/k,\ket{\psi_i}) \}_{i=0}^{k-1}$,
\begin{equation}
    \mathbf{C}_{\text{MIO}}(\Omega) + \mathbf{S}(\Omega) \leq \log_2 d,
\end{equation}
\end{theorem}

\begin{proof}
Note that for the
$\Omega = \{(1/k,\ket{\psi_i}) \}_{i=0}^{k-1}$, $\mathbf{S}(\Omega) = \mathbf{S}_{\rm min}(\Omega) = \log_2 k $, we arrive the conclusion immediately, combined with the Theorem~\ref{theo:dual_min}.
\end{proof}

This bound can be achieved in the example in Proposition \textcolor{blue}{3}.

\renewcommand\theproposition{\textcolor{blue}{3}}
\setcounter{proposition}{\arabic{proposition}-1}
\begin{proposition}
Let $\ket{\phi_i} = HX^{i}\ket{0}$ where $H, X$ are the $d$-dimensional Hadamard gate and generalized Pauli $X$ gate. For $\Omega = \{(1/k,\ket{\phi_i})\}_{i=0}^{k-1}, \, k\leq d$,
\begin{equation}
\mathbf{C}_{\text{MIO}}(\Omega) + \mathbf{S}(\Omega) = \log_2 d.
\end{equation}
\end{proposition}

\begin{proof}
    We begin by verifying that the states $\{\ket{\phi_i})\}_{i=0}^{k-1}$ are mutually orthogonal maximally coherent states. 
    The $d$-dimensional Hadamard gate $H$, Pauli $X$ gate and  Pauli $Z$ gate are defined as: $H = \frac{1}{\sqrt{d}}\sum_{i,j=0}^{d-1}\omega^{kj}\ketbra{k}{j}$, $X = \sum_{j=0}^{d-1}\ketbra{j+1}{j}$ and $Z = \sum_{j=0}^{d-1}\omega^{j}\ketbra{j}{j}$, where $\{ \ket{j}\}_{j=0}^{d-1}$ denotes computational basis and $\omega = e^{i2\pi/d}$.
    Firstly, we have the maximally coherent state $\ket{\Psi_d} = H\ket{0}$, and each $\ket{\phi_i} = HX^{i}\ket{0} = Z^{i}H\ket{0} = Z^{i}\ket{\Psi_d}$.
     Each $\ket{\phi_i}$ is also a maximally coherent state because $Z^{i}$ is a diagonal unitary~\cite{Peng_2016}. Since $\langle \phi_i | \phi_j\rangle = 0\ \forall i \neq j, i,j \leq d-1$, these states are mutually orthogonal.
    Then we construct an MIO channel $\cN$ to show that equality can be achieved in Eq.~\eqref{eq: upper_ROC} for this ensemble containing $k$ mutually orthogonal maximally coherent states. Let $\cE = \{ \ket{\phi_i} \}_{i=0}^{d-1}$ be a complete orthogonal basis of $\cH_A$ including all $d$ mutually orthogonal maximally coherent states. Note that $\ket{\phi_j} \in \cE$ but $\ket{\phi_j} \notin \Omega, j = k, \cdots, d-1$.
    We denote $ \ketbra{\phi_i}{\phi_i} $ with $\Psi_i$ and construct an MIO channel $\cN$ as follows:
     \begin{equation}
        \label{eq: up_MIO}
        \cN(\rho) = \sum_{i=0}^{k-1} \tr(\Psi_i \rho)\ketbra{i}{i} \ox \sigma + \sum_{j=k}^{d-1} \tr(\Psi_j \rho) \Pi \ox \sigma', \quad \forall \rho,
    \end{equation}
    where $\Pi = 1/k\sum_{i=0}^{k-1}\ketbra{i}{i}$, $\sigma = I/d + \sum_{m\neq n} \frac{d-k}{k(d-1)d} \ketbra{m}{n}$ and $\sigma' = I/d - \sum_{m\neq n} \frac{1}{(d-1)d} \ketbra{m}{n}$.
    First, we can check that $\cN$ is an MIO because $\cN(\ketbra{m}{m}) = \Pi \ox \frac{I}{d}$ where $m = 0, \cdots, d-1$. Then we can check that:
    \begin{equation}
        \label{eq: up_MIO_result}
         \cN(\Psi_i) =  \ketbra{i}{i} \ox \sigma, \quad i=0,\cdots, k-1.
    \end{equation}
    Thus, $\mathbf{C}_{\text{MIO}}(\Omega) \geq \sum_{i=0}^{k-1} C_{R}(\sigma)/k = (d-k)/k$. Combined with the upper bound provided by the Theorem~\ref{thm:roc_upper}, we can conclude that $\mathbf{C}_{\text{MIO}}(\Omega) = (d-k)/k$ and
    $\mathbf{C}_{\text{MIO}}(\Omega)  =  \log_2 (d/k) = \log_2 d - \mathbf{S}(\Omega)$.
\end{proof}

\renewcommand\theproposition{\textcolor{blue}{4}}
\setcounter{proposition}{\arabic{proposition}-1}

\renewcommand{\theproposition}{S\arabic{proposition}}
 \begin{proposition}
    For a mutually orthogonal $d$-dimensional pure-state ensemble $\Omega = \{(1/k,\ket{\psi_i}) \}_{i=0}^{k-1}$ with $k <d$, the necessary condition for $\Omega$ to saturate the coherence-distinguishability duality relation is 
    \begin{equation}
        C_{\rm max}(\hat{\omega}) + S(\hat{\omega}) = \log_2 d,
    \end{equation}
      where $\hat{\omega} =  1/k\sum_{i=0}^{k-1} \ketbra{\psi_i}{\psi_i}$.
\end{proposition}
\begin{proof}
    Saturating the coherence-distinguishability duality relation with a given $\Omega$ means $\mathbf{C}_{\text{MIO}}(\Omega) + \mathbf{S}(\Omega) = \log_2 d$. Note that we restrict $k < d$ and avoid the trivial case of extracting all $\log_2 d$ c-bits from $\Omega$, namely $\mathbf{S}(\Omega) = \log_2 d$.
    Suppose the MIO achieving $\mathbf{C}_{\text{MIO}}(\Omega)$ is $\cN$, combined with Lemma~\ref{lem: roc_ensemble}, if
   $ C_{\rm max}(\hat{\omega}) + S(\hat{\omega}) < \log_2 d$, we have
   \begin{align}
       \mathbf{C}_{\text{MIO}}(\Omega) + \mathbf{S}(\Omega) & = C_{\rm max}(\cN(\hat{\omega})) + \mathbf{S}(\Omega) \\
       & = C_{\rm max}(\cN(\hat{\omega})) + S(\hat{\omega}) \\
       & \leq C_{\rm max}(\hat{\omega}) + S(\hat{\omega}) \\
       & < \log_2 d.
   \end{align}
    Then it is impossible for $\Omega$ to saturate the coherence-distinguishability duality relation. Then we have $C_{\rm max}(\hat{\omega}) + S(\hat{\omega}) \geq \log_2 d$ is the necessary condition.
    Besides, we have:
    \begin{align}
       \frac{I}{d} & = \frac{1}{d} [\sum_{i=0}^{k-1} \ketbra{\psi_i}{\psi_i} + \sum_{j=k}^{d-1} \ketbra{\psi_j}{\psi_j}] \\
       & = \frac{k}{d} (\hat{\omega} + \frac{d-k}{k}\tau),
   \end{align}
    where $\tau = \frac{1}{d-k}\sum_{j=k}^{d-1} \ketbra{\psi_j}{\psi_j}$. We have $C_{\rm max}(\hat{\omega}) \leq \log_2 (1+ \frac{d-k}{k}) = \log_2 d - \log_2 k$. That is $C_{\rm max}(\hat{\omega}) + S(\hat{\omega}) \leq \log_2 d$. Then we conclude that $C_{\rm max}(\hat{\omega}) + S(\hat{\omega}) = \log_2 d$ is the necessary condition.
\end{proof}

\end{document}